\documentclass[conference]{ieeeconf}  

\IEEEoverridecommandlockouts

\usepackage{graphicx} 
\usepackage{amsmath} 
\usepackage{array}   
\usepackage{booktabs} 
\usepackage{caption} 
\usepackage{siunitx} 
\usepackage{epstopdf} 

\usepackage{cite}
\usepackage{amsmath,amssymb,amsfonts,mathtools,nccmath,bm, tikz, enumerate}
\usepackage{multicol}
\usepackage[nodisplayskipstretch]{setspace}
\newtheorem{definition}{Definition}[section] 
\newtheorem{assumption}{Assumption}[section]
\newtheorem{remark}{Remark}[section]
\newtheorem{theorem}{Theorem}
\newtheorem{lemma}{Lemma}
\newtheorem{proposition}{Proposition}[section]
\newtheorem{corollary}{Corollary}[section]
\usepackage{graphicx}
\usepackage{textcomp}
\usepackage{cuted}
\usepackage{float}
\let\labelindent\relax
\usepackage{enumitem}
\usepackage{xcolor}
\usepackage{algorithm}
\usepackage{algpseudocode}

\def\BibTeX{{\rm B\kern-.05em{\sc i\kern-.025em b}\kern-.08em
    T\kern-.1667em\lower.7ex\hbox{E}\kern-.125emX}}
\allowdisplaybreaks

\newcommand{\Uset}{\mathcal{U}}
\newcommand{\Ut}{\tilde{\mathcal{U}}}

\usepackage{tikz}
\usepackage{subcaption}
\usetikzlibrary{arrows.meta}

\usetikzlibrary{positioning,arrows.meta,shapes.geometric,fit,backgrounds,calc}

\definecolor{bluelight}{HTML}{E7EFFB}   \definecolor{bluedark}{HTML}{2C5AA0}
\definecolor{greenlight}{HTML}{EAF7EC}  \definecolor{greendark}{HTML}{3A8451}
\definecolor{orangelight}{HTML}{FDF1DE} \definecolor{orangedark}{HTML}{D08A2A}
\definecolor{purplelight}{HTML}{EFEAFA} \definecolor{purpledark}{HTML}{6C4FA0}
\definecolor{redlight}{HTML}{FBE7E6}    \definecolor{reddark}{HTML}{B23A32}
\definecolor{graylight}{HTML}{F5F5F5}   \definecolor{graydark}{HTML}{9AA0A6}

\tikzset{
  reqbox/.style   ={rectangle,rounded corners=3pt,draw=bluedark,fill=bluelight,
                     text width=3.6cm,align=left,inner sep=6pt,font=\footnotesize},
  propbox/.style  ={rectangle,rounded corners=3pt,draw=greendark,fill=greenlight,
                     text width=3.6cm,align=left,inner sep=6pt,font=\footnotesize},
  certbox/.style  ={rectangle,rounded corners=3pt,draw=orangedark,fill=orangelight,
                     text width=3.4cm,align=left,inner sep=6pt,font=\footnotesize},
  outbox/.style   ={rectangle,rounded corners=3pt,draw=purpledark,fill=purplelight,
                     text width=3.4cm,align=left,inner sep=6pt,font=\footnotesize},
  failbox/.style  ={rectangle,rounded corners=3pt,draw=reddark,fill=redlight,
                     text width=3.0cm,align=left,inner sep=6pt,font=\footnotesize},
  decision/.style ={diamond,aspect=2.2,draw=orangedark,fill=orangelight,
                     align=center,inner sep=1pt,font=\footnotesize\bfseries,
                     text width=1.5cm},
  groupbg/.style  ={rectangle,rounded corners=5pt,draw=graydark,fill=graylight,
                     inner sep=10pt},
  arr/.style      ={-{Latex[length=2mm]},thick,draw=black!70},
  yarr/.style     ={-{Latex[length=2mm]},thick,draw=greendark},
  narr/.style     ={-{Latex[length=2mm]},thick,dashed,draw=reddark},
  loopc/.style     ={-{Latex[length=2mm]},thick,dashed,draw=purpledark},
}

\begin{document}

\title{\textbf{Scalable Tube-Tightened Multi-Agent Safety
via Certified Constraint Reduction} 
}

\author{A, \IEEEmembership{Member} 
\thanks{A.with the department of \mbox{Mechanical} }
}

\author{Armel Koulong, \IEEEmembership{Member, IEEE}
\thanks{A. Koulong is with the department of \mbox{Mechanical} Engineering, University of Alabama, Tuscaloosa, AL, USA (e-mail: akoulongzoyem@crimson.ua.edu). }
}

\maketitle
\thispagestyle{plain}
\pagestyle{plain}

\begin{abstract}
This paper develops a certified constraint-reduction method for distributed
model predictive control with tube-tightened exponential control barrier
functions (eCBFs) in multi-agent systems. At each prediction stage, pairwise
agent--agent and agent--obstacle eCBF conditions define halfspaces in the local
control space. Rather than enforcing all such halfspaces, a geometry-adaptive
subset is retained and a Farkas certificate verifies that the reduced
admissible set is contained in the full tightened set. For planar inputs, cone
coverage is characterized through the largest angular gap: two extreme
directions suffice in the strict half-plane regime, while other geometries
initialize with three retained constraints and escalate only when
certification fails. Conic multipliers and nominal-aware offsets are obtained
in closed form, without an auxiliary optimization, and the resulting
construction preserves any nominal control already admissible for the full
tightened set. Consequently, the reduced controller inherits the robust safety
guarantee of the underlying tube-eCBF formulation. In a ten-follower,
four-obstacle study, the method retained fewer safety constraints on
average, reproduced the full filter's nominal accept/reject decisions with no
true safety violations, and achieved increasing computational gains as the
constraint count and prediction horizon grew.
\end{abstract} 

\section{Introduction} \label{Introduction}

Distributed model predictive control (DMPC) provides a natural architecture
for constrained multi-agent coordination -- each agent solves a local
finite-horizon optimal control problem while coupling with its neighbors
through exchanged predicted trajectories
\cite{dunbar2006dmpc,dai2017dmpc}. Its computational advantage over
centralized planning, however, can erode in dense environments. Collision
avoidance and obstacle avoidance introduce one or more pairwise safety
constraints for every relevant interaction and prediction stage, so the local
optimization burden grows with both neighborhood density and horizon length \cite{koulong2026cdc}.

Control barrier functions (CBFs) provide a systematic mechanism for enforcing
such safety requirements by converting forward-invariance conditions into
inequalities on the control input \cite{ames2017cbf7}. Exponential and
high-order CBFs extend this construction to constraints of higher relative
degree \cite{xiao2022hocbf}, and barrier certificates have been
used extensively for collision avoidance in multi-agent systems
\cite{koulong2025wc}. More recent work has developed distributed algorithms
for solving CBF-constrained network optimization problems
\cite{mestres2024navigation}, as well
as graph-based barrier representations aimed at large-scale multi-agent
systems \cite{zhang2025gcbf}.

CBFs have also been incorporated directly into predictive multi-robot
controllers. Distributed MPC formulations with collision and obstacle
constraints have long been studied for formation control
\cite{dai2017dmpc}, while recent methods combine DMPC with discrete-time CBFs
or high-order CBFs to obtain predictive safety guarantees
\cite{jiang2024cbfdmpc,wang2025dscmpc}. In the presence of disturbances,
tube-based robust MPC provides a standard means of separating nominal planning
from bounded tracking error \cite{mayne2005robust}. Building on this idea,
\cite{koulong2026cdc} derives support-function tightenings of high-relative-
degree eCBF constraints from robust positively invariant tubes, so that
feasibility of a nominal distributed MPC plan guarantees safety of the true
disturbed multi-agent system.

These constructions strengthen the safety guarantees of distributed
predictive control but leave a scalability issue: every relevant pairwise eCBF
inequality is still imposed explicitly along the prediction horizon. Recent
multi-agent CBF methods address related computational burdens by activating
safety constraints only for neighbors that are sufficiently close or
approaching \cite{su2026motion}, or by allocating pairwise collision-avoidance
responsibility among agents \cite{autenrieb2025auction}. Such methods reduce
the number of active constraints through interaction relevance or assignment.
They do not, however, ask whether one safety inequality is already implied by
other retained inequalities.

Constraint reduction has been studied separately in the MPC literature.
Online constraint-removal methods identify inequalities that are guaranteed to
be inactive and omit them without changing the optimizer
\cite{jost2015removal}. Constraint-adaptive
MPC extends this idea using reachability and optimality arguments to recover
the original MPC law while imposing only a state-dependent subset of
constraints \cite{nouwens2023cammpc}. Constraint aggregation instead replaces
many nonlinear inequalities by a smaller approximate representation of the
feasible region \cite{pereira2022aggregation}. These approaches motivate
reducing the optimization burden before the solver is called, but their
selection criteria are based on inactivity, optimality, reachability, or
approximation rather than logical implication among safety halfspaces.

This paper develops a certified reduction principle for tube-tightened
multi-agent eCBFs based on implication among the safety halfspaces themselves.
At a fixed predicted state, a geometry-adaptive subset of the pairwise eCBF
halfspaces is retained, and a Farkas containment certificate verifies that its
intersection is contained in the full admissible-control set. Thus a constraint
may be removed even if it is safety-relevant or potentially active, provided it
is implied by the retained constraints. To the best of the authors' knowledge,
existing multi-agent CBF reduction methods do not exploit such certified
halfspace implication for tube-tightened pairwise eCBFs. For planar inputs, the
directional certificate reduces to cone coverage -- a largest-angular-gap rule
retains two extreme directions in the strict half-plane regime and otherwise
initializes with three directions and escalates only if certification fails.
Conic Carath\'eodory then yields closed-form multipliers through at most a
$2\times2$ solve, while the reduced offsets admit a closed-form nominal-aware
construction.

The main contributions are:
\begin{enumerate}
\item A Farkas-certified inner reduction of the tube-tightened pairwise eCBF
admissible-control set, guaranteeing
$\widetilde{\mathcal U}_i\subseteq\mathcal U_i$ independently of the rule used
to propose the retained subset.

\item A planar geometry-adaptive construction based on cone coverage, with
closed-form conic multipliers and nominal-aware offsets; the construction
retains two extreme directions in the strict half-plane regime and otherwise
initializes with three constraints and escalates only when certification fails.

\item A robust-safety guarantee together with preservation of any prescribed
nominal control that is already admissible for the full tightened set, and a
stagewise DMPC implementation that reuses retained indices, re-certifies them
against the evolving geometry, and recovers the unreduced problem as its
full-set fallback.
\end{enumerate}

\section{Problem Formulation} \label{ProblemFormulation}

We consider $N$ follower agents $i \in \mathcal{V} = \{1,\ldots,N\}$ and a
leader $0$. Each follower has $n$-th order nonlinear Brunovsky dynamics
\begin{align}\label{eq:ct_follower}
\dot{x}_p^{i} &= x_{p+1}^{i}, &p \in \{1,\dots,n-1\}, \notag\\
\dot{x}_n^{i} &= f^{i}(x^{i},t) + u^{i} + w^{i}, &w^{i}(t)\in\mathcal D_i,
\end{align}
where $x_p^i\in\mathbb{R}^d$, $u^i\in\mathbb{R}^d$ is the input, and
$w^i\in\mathbb{R}^d$ is a bounded time-varying disturbance with
$\|w^{i}(t)\|\le \bar w^i$. In control-affine form
$\dot x^{\,i} = F_x(x^{\,i},t) + F_u u^{\,i} + F_w w^{\,i}$,
\begin{align}\label{eqn:controlaffine}
  F_x(x^{\,i},t) = \begin{bmatrix} x_2^{\,i}\\ \vdots \\ x_n^{\,i}\\ f^{\,i}(x^{\,i},t)\end{bmatrix},
  \qquad
  F_u = F_w = \begin{bmatrix} 0\\ \vdots\\ 0\\ I_d\end{bmatrix}.
\end{align}
The leader dynamics satisfy
\begin{align}\label{eq:ct_leader}
\dot{x}_p^{\,0} &= x_{p+1}^{\,0},\hspace{85pt} p \in \{1,\dots,n-1\}, \notag\\
\dot{x}_n^{\,0} &= f^{\,0}(x^{\,0},t) + w^{0},\hspace{40pt} w^{0}(t)\in\mathcal D_0,
\end{align}
where $w^0\in\mathbb{R}^d$ is a bounded time-varying disturbance of the leader.
Both $f^i$ and $f^0$ are locally Lipschitz with constants $L_i$ and $L_0$ over a
compact domain of interest.

The agents communicate over a fixed weighted graph
$\mathcal{G} = (\mathcal{V},\mathcal{\phi},[A]_{ij})$ with weighted adjacency $[A]_{ij} >0$ iff $(j,i)\in\mathcal{\phi}$, and Laplacian $L:=[D]_{ij}-[A]_{ij}$. The neighbors
set is $\mathcal{N}_i=\{j: [A]_{ij}>0\}$; each agent has perpetual access to its
own state $x^i(t)\in\mathbb{R}^{nd}$ and to $x^j(t)$ for every
$j\in\mathcal{N}_i$. The augmented graph
$\bar{\mathcal{G}}=(\bar{\mathcal{V}},\bar{\mathcal{\phi}})$ with
$\bar{\mathcal{V}}=\{0,1,\dots,N\}$ adds the leader interaction matrix
$B_0=\mathrm{diag}\{[B]_{i0}\}$.

\begin{assumption}\label{ass:graph}
The augmented graph $\bar{\mathcal{G}}$ with node set $\{0\}\cup\mathcal{V}$
contains a spanning tree rooted at the leader.
\end{assumption}

Assumption~\ref{ass:graph} is a design-stage condition verified before
deployment; it implies $\nu_1 L+\nu_2 B_0$ is a nonsingular
M-matrix~\cite{koulong2025acc} for any $\nu_1,\nu_2>0$.


Each follower must drive the synchronization error
$(x^i_p-x^0_p-\psi^i_p)$ to zero under bounded disturbances and hard safety
constraints. Under the feedforward ancillary law of~\cite{koulong2026cdc},
\begin{equation}
  u^i(t)=\bar{v}^i(t_k)-K^i_p\delta x^i(t)+f^0(\hat x^0,t)-f^i(x^i,t),
  \label{eq:control}
\end{equation}
the tracking error $\delta x^i:=x^i-\bar x^i$ is confined for all admissible
disturbances to the robust positively invariant ellipsoid
\begin{equation}\label{eq:tube}
  \mathcal{Z}^i:=\{z\in\mathbb{R}^{nd}: z^\top P_i z\le r_i^2\},
\end{equation}
with $P_i\succ0$ and radius $r_i$ given in closed form
by~\cite[Proposition~1]{koulong2026cdc}. Here $\bar{v}^i\in\mathcal{V}^i:=\mathcal{U}^i\ominus K^i_p\mathcal{Z}^i$
is an optimal control input held constant over $[t_k,t_{k+1})$, while
$-K^i_p\delta x^i(t)$ and the feedforward act continuously. Note that $r_i$ is a
worst-case invariant radius and is therefore computed once offline,
whereas the margins $\delta_{ij}$ below depend on the barrier gradient and are
evaluated online at every step.

\subsection{Tightened eCBF Functions \cite{koulong2026cdc}}
Safety is encoded by the pairwise and obstacle barrier functions
\begin{equation}\label{eq:barrier}
  h_{ij}(\bar x^i,\bar x^j)=\|\bar x^i_1-\bar x^j_1\|^2-D^2,
  \;\;
  h_{iO}(\bar x^i)=\|\bar x^i_1-c_O\|^2-R_O^2,
\end{equation}
with $D$ the minimum inter-agent separation and $(c_O,R_O)$ the center and
effective radius of obstacle $O$.

\begin{assumption}\label{ass:reldeg}
Each $h_{ij}$ and $h_{iO}$ is convex, $r$ times continuously differentiable, and
has uniform relative degree $r$ with respect to \eqref{eqn:controlaffine}. The
gains $\kappa_0,\dots,\kappa_{r-1}$ are chosen so that
$s^{r}+\kappa_{r-1}s^{r-1}+\cdots+\kappa_0$ is Hurwitz.
\end{assumption}

For the barriers \eqref{eq:barrier} under \eqref{eqn:controlaffine} the relative
degree is $r=n$, since $u^i$ enters only at the $n$-th derivative.
The Lie derivatives along $F_x$ are defined recursively by
$L_{F_x}^0h_\bullet:=h_\bullet$ and
$L_{F_x}^{q}h_\bullet:=\nabla(L_{F_x}^{q-1}h_\bullet)^\top F_x$, and along $F_u$
by $L_{F_u}h_\bullet:=\nabla h_\bullet^\top F_u$. The support-function
tightening mechanism of~\cite[Lemma 2.3]{koulong2025wc} is applied to both
inter-agent and obstacle avoidance:
\begin{multline}\label{eq:eCBF_col_final}
\hspace{-9pt} \Phi_{ij}^{\text{tight}}\big(\bar{x}^{i}, \bar{x}^{j}, \bar{v}^{i}\big)
 =L_{F_{x}}^{r} h_{ij}(\bar{x}^{i},\bar{x}^{j})
 + \sum_{q=1}^{r-1} \kappa_q L_{F_{x}}^{q} h_{ij}(\bar{x}^{i},\bar{x}^{j}) \\
 \hspace{9pt} + \big(L_{F_{u}} L_{F_{x}}^{r-1} h_{ij}(\bar{x}^{i},\bar{x}^{j})\big) \bar{v}^{i}
+\kappa_0 \big( h_{ij} (\bar x^i,\bar x^j) -\delta_{ij}(\bar{x}^{i},\bar{x}^{j})\big),
\end{multline}
with the obstacle counterpart $\Phi_{iO}^{\text{tight}}(\bar{x}^{i},\bar{v}^{i})$
defined identically by replacing $h_{ij}$ with $h_{iO}$ and $\delta_{ij}$ with
$\delta_{iO}$. The margins are the exact ellipsoidal support functions of the
tubes $\mathcal Z^i$ along the barrier gradients,
\begin{align}\label{eq:margins}
\delta_{ij} &:= r_i\sqrt{\nabla_{\bar{x}^i}h_{ij}^\top P_i^{-1}\nabla_{\bar{x}^i}h_{ij}}
             + r_j\sqrt{\nabla_{\bar{x}^j}h_{ij}^\top P_j^{-1}\nabla_{\bar{x}^j}h_{ij}}, \notag\\
\delta_{iO} &:= r_i\sqrt{\nabla_{\bar{x}^i}h_{iO}^\top P_i^{-1}\nabla_{\bar{x}^i}h_{iO}} .
\end{align}

\subsection{Halfspace form}
Only the term $\big(L_{F_u}L_{F_x}^{r-1}h_{ij}\big)\bar v^{\,i}$ in
\eqref{eq:eCBF_col_final} contains the decision variable; every other quantity is
a number once $\bar x^i$ and the neighbor's predicted trajectory are fixed.
Imposing $\Phi_{ij}^{\text{tight}}\ge0$ and isolating $\bar v^{\,i}$ therefore
gives the halfspace
\begin{equation}\label{eq:lin}
  a_{ij}^\top \bar v^{\,i} \;\ge\; b_{ij},
\end{equation}
with
\begin{equation}
  a_{ij} := \big(L_{F_u}L_{F_x}^{r-1}h_{ij}(\bar x^{i},\bar x^{j})\big)^{\!\top},
  \label{eq:aij}
\end{equation}
\begin{multline}
  b_{ij} := -\Big[\,L_{F_x}^{r}h_{ij}(\bar x^{i},\bar x^{j})
            + \sum_{q=1}^{r-1}\kappa_q L_{F_x}^{q}h_{ij}(\bar x^{i},\bar x^{j})
           \\ + \kappa_0\big(h_{ij}(\bar x^{i},\bar x^{j})-\delta_{ij}(\bar x^{i},\bar x^{j})\big)\Big].
  \label{eq:bij}
\end{multline}
The obstacle constraint $\Phi_{iO}^{\text{tight}}\ge0$ yields
$a_{iO}^\top\bar v^{\,i}\ge b_{iO}$ by the same substitution, replacing
$h_{ij}$ with $h_{iO}$ and $\delta_{ij}$ with $\delta_{iO}$.

\begin{remark}\label{rem:pivot}
Along the horizon, agent $j$'s predicted trajectory is generated by the plan
$\hat v^{\,j}$ broadcast at the previous instant, so
$L_{F_x}^{r}h_{ij}$ carries the known term $-a_{ij}^\top\hat v^{\,j}$. Together
with the internal nonlinearities $f^{i},f^{j}$ and the tube margin
$\delta_{ij}$, every neighbors and disturbance-dependent quantity enters
\eqref{eq:lin} through the scalar $b_{ij}$ and never through the
vector $a_{ij}$. Every neighbors-dependent quantity therefore sits inside the scalar
$b_{ij}$, and none of them touches the vector $a_{ij}$.
\end{remark}

\subsubsection*{Problem Statement}
Each agent must avoid every neighbor and every obstacle, producing one
tightened constraint of the form \eqref{eq:eCBF_col_final} per threat. Inside a distributed MPC, the optimization does not return a single input -- it returns a sequence
$\{\bar v^i(l)\}_{l=0}^{N_p-1}$, and \eqref{eq:eCBF_col_final} must hold at
every predicted step, not only the first. Each constraint is therefore
imposed $N_p$ times, and the program carries
\begin{equation}\label{eq:growth}
  N_p\big(|\mathcal{N}_i| + N_{\mathrm{obs}}\big)
\end{equation}
inequalities per solve, where $N_p$ is the prediction horizon length. Since the
cost of a quadratic program grows superlinearly in its constraint count \cite{boyd2004convex}, and
since \eqref{eq:growth} scales with both neighborshood size and horizon length,
this is the dominant computational burden of the scheme in dense formations or
over long horizons.
Rather than enforcing every tightened inequality explicitly, we ask whether a
small, geometry-adaptive subset can be enforced instead, with a certificate that
the resulting admissible set is contained in the full one.

\section{Methodology} \label{Methodology}

Let $\mathcal O_i$ denote the set of obstacle-induced safety constraints
considered by agent $i$, and define
$\mathcal J_i := \mathcal N_i\cup\mathcal O_i.
$
Thus $\mathcal J_i$ indexes all eCBF halfspaces participating in the
constraint reduction, including both neighboring agents and relevant
obstacles.
We define
$ \Uset_i
    :=
    \left\{
        u\in\mathbb U:
        a_{ij}^{\top}u\ge b_{ij},
        \ \forall j\in\mathcal J_i
    \right\},
$
and, for a retained index set $\mathcal K_i\subseteq\mathcal J_i$,
$\Ut_i
    :=
    \left\{
        u\in\mathbb U:
        \tilde a_{ik}^{\top}u\ge\tilde b_{ik},
        \ \forall k\in\mathcal K_i
    \right\}.$
Safety is preserved whenever $\Ut_i\subseteq\Uset_i$. For compactness, we denote the directional tube margin by
$\delta(a) := \sigma_i(Ea)+\sigma_j(-Ea), \text{with } E
:=
[\,I_d\;0\;0\,]^\top
\in\mathbb R^{nd\times d},$
so that, for the pairwise constraint evaluated at
$(\bar x^i,\bar x^j)$, $ \delta_{ij}(\bar x^i,\bar x^j) =
\delta\!\left(a_{ij}(\bar x^i,\bar x^j)\right).$
For obstacle constraints, the second support-function term is omitted.

\begin{lemma}[Containment certificate]\label{lem:farkas}
If for every neighbor $j\in\mathcal J_i$ there exists
$\lambda_j\in\mathbb R_{\ge0}^{|\mathcal K_i|}$ satisfying
$ a_{ij}
    =
    \sum_{k\in\mathcal K_i}
    \lambda_{jk}\tilde a_{ik},
    \;
    \sum_{k\in\mathcal K_i}
    \lambda_{jk}\tilde b_{ik}
    \ge b_{ij}, $
then
$ \Ut_i\subseteq\Uset_i.$
\end{lemma}

\begin{proof}
Let $u\in\Ut_i$, so $\tilde a_{ik}^\top u\ge\tilde b_{ik}$ for all $k$. Since
$\lambda_{jk}\ge0$, multiplying preserves each inequality; summing and applying
the two conditions gives $a_{ij}^\top u\ge b_{ij}$. As $j$ was arbitrary,
$u\in\Uset_i$.
\end{proof}

From the offset condition of Lemma~\ref{lem:farkas}, we decompose the full and
retained offsets into nominal and tube-tightening terms:
\begin{align*}
    b_{ij}
    &=
    \underbrace{b_{ij}^{\mathrm{nom}}}_{\text{drift, plan, eCBF stack}}
    +
    \underbrace{\kappa_0\delta(a_{ij})}_{\text{tube margin}},\quad \tilde b_{ik}
    =
    \tilde b_{ik}^{\mathrm{nom}}
    +
    \kappa_0\delta(\tilde a_{ik}).
\end{align*}
The offset condition
$\sum_k\lambda_{jk}\tilde b_{ik}\ge b_{ij}$ then becomes
\begin{equation}\label{eq:split}
  \underbrace{\sum_k\lambda_{jk}\tilde b_{ik}^{\mathrm{nom}}}
             _{\text{nominal supply}}
  +
  \underbrace{\sum_k\lambda_{jk}\kappa_0\delta(\tilde a_{ik})}
             _{\text{margin supply}}
  \ge
  \underbrace{b_{ij}^{\mathrm{nom}}}_{\text{nominal demand}}
  +
  \underbrace{\kappa_0\delta(a_{ij})}_{\text{margin demand}} .
\end{equation}
Tube tightening increases $b_{ij}$ and hence makes the eCBF constraint more
restrictive. The following lemma shows that, under the directional cone
certificate, these robustness margins compose in the safe direction required
for containment.

\begin{lemma}[Tube-tightening certificate transfer]
\label{lem:tube-transfer}

Let $\kappa_0\ge0$ and suppose
$
    a_{ij}
    =
    \sum_{k\in\mathcal K_i}
    \lambda_{jk}\tilde a_{ik},
    \;
    \lambda_{jk}\ge0.
$
Then
$
    \sum_{k\in\mathcal K_i}
    \lambda_{jk}\kappa_0\delta(\tilde a_{ik})
    \ge
    \kappa_0\delta(a_{ij}).
$
Consequently, any nominal offset certificate
$
    \sum_{k\in\mathcal K_i}
    \lambda_{jk}\tilde b^{\mathrm{nom}}_{ik}
    \ge b^{\mathrm{nom}}_{ij}
$
is preserved after tube tightening:
$
    \tilde b_{ik}
    =
    \tilde b^{\mathrm{nom}}_{ik}
    +\kappa_0\delta(\tilde a_{ik}),
    \;
    b_{ij}
    =
    b^{\mathrm{nom}}_{ij}
    +\kappa_0\delta(a_{ij})
$
implying
$
    \sum_{k\in\mathcal K_i}
    \lambda_{jk}\tilde b_{ik}
    \ge b_{ij}.
$
\end{lemma}

\begin{proof}
For the ellipsoidal tube
$\mathcal Z_i=\{e:e^\top P_i e\le r_i^2\}$,
$
    \sigma_i(Ea)
    =
    r_i\sqrt{a^\top E^\top P_i^{-1}Ea}.
$
Hence each support-function term is positively homogeneous and subadditive,
and so is
$
    \delta(a)=\sigma_i(Ea)+\sigma_j(-Ea).
$
Using $\lambda_{jk}\ge0$ and the directional certificate,
\begin{align*}
    \sum_{k\in\mathcal K_i}
        \lambda_{jk}\delta(\tilde a_{ik})
    &=
    \sum_{k\in\mathcal K_i}
        \delta(\lambda_{jk}\tilde a_{ik}) \\
    &\ge
    \delta\!\left(
        \sum_{k\in\mathcal K_i}
        \lambda_{jk}\tilde a_{ik}
    \right)
    =
    \delta(a_{ij}).
\end{align*}
Multiplication by $\kappa_0\ge0$ proves the first claim. Therefore,
\begin{align*}
    \sum_{k\in\mathcal K_i}\lambda_{jk}\tilde b_{ik}
    &=
    \sum_{k\in\mathcal K_i}
        \lambda_{jk}\tilde b^{\mathrm{nom}}_{ik}
    +
    \kappa_0
    \sum_{k\in\mathcal K_i}
        \lambda_{jk}\delta(\tilde a_{ik}) \\
    &\ge
    b^{\mathrm{nom}}_{ij}
    +
    \kappa_0\delta(a_{ij})
    =
    b_{ij}.
\end{align*}
\end{proof}

Lemma~\ref{lem:tube-transfer} shows that tube tightening preserves a valid
nominal containment certificate: the same conic multipliers remain valid after
the robustness margins are added. Thus tube tightening does not invalidate the
constraint-reduction certificate. In the implementation, the reduction is
applied directly to the already tightened offsets $b_{ij}$.

\subsection{Selecting the directions}
\label{sec:cone}

\begin{definition}[Cone Membership]
\label{def:cone_mem}
Given retained directions
$\{\tilde a_{i1},\ldots,\tilde a_{iK}\}\subset\mathbb R^d$,
an original direction $a_{ij}$ satisfies the cone-membership condition if
\begin{equation}\label{eq:cone}
    a_{ij}
    \in
    \operatorname{cone}
    \{\tilde a_{i1},\ldots,\tilde a_{iK}\},
\end{equation}
or, equivalently, if there exist coefficients $\lambda_{jk}\ge0$ such that
$
    a_{ij}
    =
    \sum_{k=1}^{K}
    \lambda_{jk}\tilde a_{ik}.
$
For the containment certificate, \eqref{eq:cone} is required for every
original direction $a_{ij}$, $j\in\mathcal J_i$, as illustrated in
Figure~\ref{fig:cone-vs-positive-span}a.
\end{definition}

\begin{definition}[Positive Spanning]
\label{def:pos_span}
The retained directions positively span $\mathbb R^d$ if
$
    \operatorname{cone}
    \{\tilde a_{i1},\ldots,\tilde a_{iK}\}
    =
    \mathbb R^d,
$
as illustrated in Figure~\ref{fig:cone-vs-positive-span}b.
Equivalently, every $x\in\mathbb R^d$ admits a representation
$
    x
    =
    \sum_{k=1}^{K}
    \mu_k\tilde a_{ik},
    \;
    \mu_k\ge0.
$
\end{definition}
Thus positive spanning requires conic coverage of the entire ambient space,
whereas cone membership requires coverage only of the particular original
directions $\{a_{ij}\}_{j\in\mathcal J_i}$.

\begin{figure}[htbp]
\centering
\includegraphics[width=3.40in]{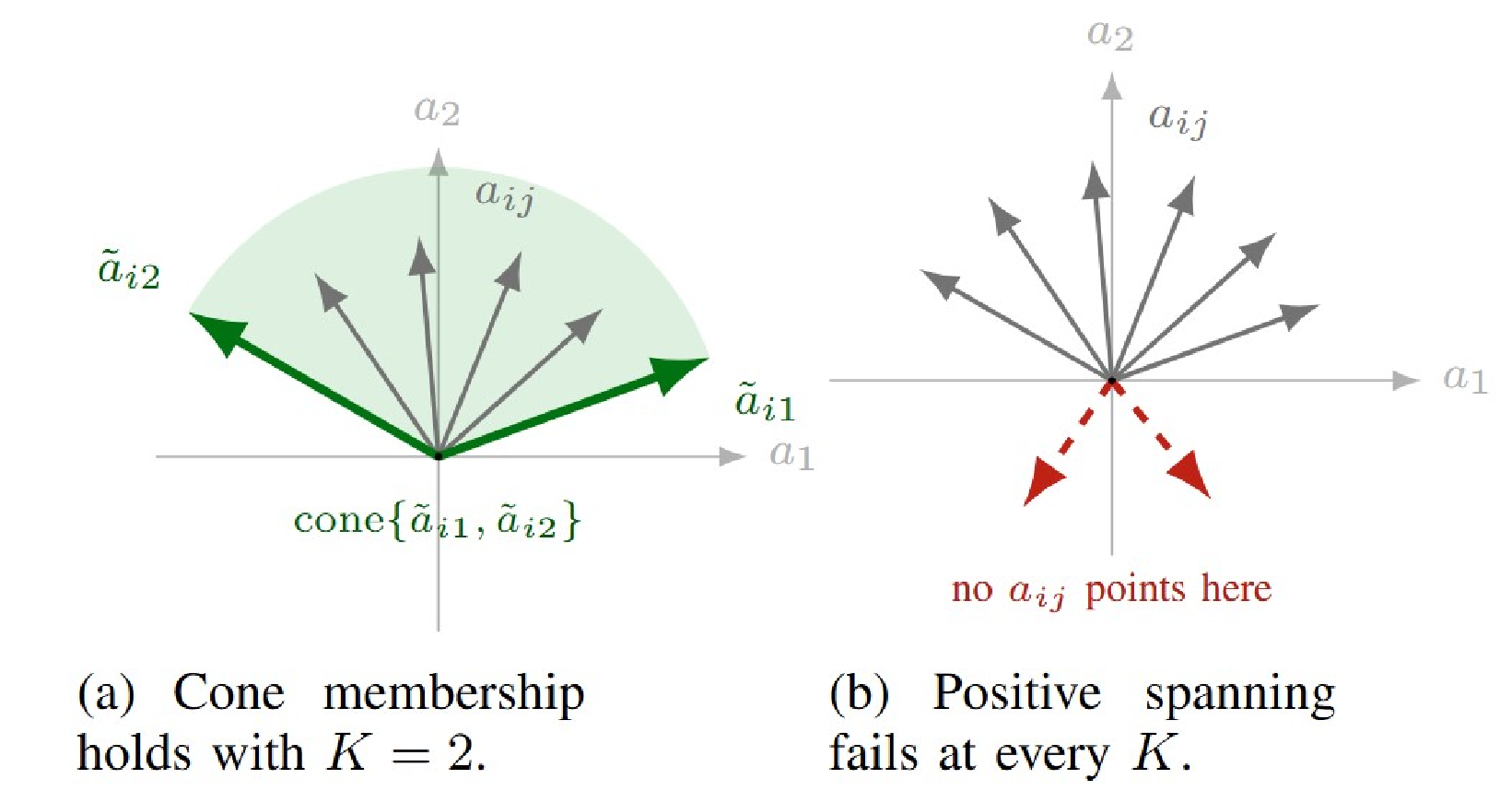}
\caption{Both panels show the same set of normals $\{a_{ij}\}$, clustered within
an arc of width less than $\pi$ as is typical for an agent on the boundary of a
formation. (a) The two extreme rays generate every normal, so \eqref{eq:cone}
holds with $K=2$. (b) Positive spanning would additionally require directions in
the lower half-plane (red dashed lines); since no $a_{ij}$ points there, no retained
subset can supply them and the condition fails for every $K$.}
\label{fig:cone-vs-positive-span}
\end{figure}

Condition~\eqref{eq:cone} is strictly weaker than requiring the retained
directions to positively span $\mathbb R^d$. Indeed, positive spanning implies
\eqref{eq:cone}, whereas the converse need not hold. Since
Lemma~\ref{lem:farkas} requires only cone coverage, positive spanning would
unnecessarily exclude valid reductions. For the distance-based barriers
considered here, $a_{ij}$ points away from neighbor $j$ (and analogously from
an obstacle). Hence, for boundary agents, the normals often occupy only a
limited angular sector. The full set may then fail to positively span
$\mathbb R^2$, in which case no retained subset can do so, while
\eqref{eq:cone} may still hold. The eCBF-halfspace polyhedron may also be
unbounded in such configurations; boundedness is therefore not the relevant
requirement. The reduction requires only
$
a_{ij}\in\operatorname{cone}\{\tilde a_{ik}\}_{k\in\mathcal K_i},
\; j\in\mathcal J_i.
$
In $d=2$, sorting the nonzero normals by angle and defining $G$ as the largest
gap between consecutive directions, including wrap-around, identifies the
corresponding cone geometry and guides the retained-set construction.

\begin{proposition}[Half-plane regime]\label{prop:cone2}
If \(G>\pi\), all normals lie within the complementary arc of width \(
2\pi-G<\pi
\). The cone they generate is therefore generated by the two extreme rays of that arc, namely the two directions adjacent to the largest gap. Retaining these two directions satisfies \eqref{eq:cone} for every $j$, so $K=2$ suffices.
\end{proposition}

\begin{proof}
A gap \(G>\pi\) leaves all normals within the complementary arc of width
\(
2\pi-G<\pi.
\)
A convex cone in \(\mathbb R^2\) whose generators lie within an angular sector of width less than \(\pi\) is generated by its two extreme rays. Hence every original normal is a nonnegative combination of the two directions adjacent to the largest gap, and \eqref{eq:cone} follows.
\end{proof}

The boundary case $G=\pi$ requires separate interpretation. In this case all
normals lie in a closed half-plane. If they are collinear, their conic hull is
a line and two antiparallel rays are mathematically sufficient. Otherwise,
their conic hull is a closed half-plane. Since the retained directions are
selected from the original set, any retained subset satisfying
\eqref{eq:cone} must generate the same conic hull as the full set. Thus two
directions cannot suffice in the non-collinear boundary case: two rays
generate either a pointed cone or a line, whereas the full cone is a
half-plane.

For a uniform numerical implementation, the entire $G\le\pi$ branch is
initialized with $K=3$ whenever at least three constraints are present; the
exceptional collinear case is therefore handled conservatively without a
separate collinearity test.

Following Proposition~\ref{prop:cone2}, the strict half-plane regime is completely resolved: two extreme rays are sufficient to generate every original normal. Outside this regime, however, a two-direction representation is no longer guaranteed. The problem therefore shifts from identifying the extreme rays of a single cone to finding a small subset whose conic hull covers the full set of normals. In \(\mathbb R^2\), the natural next candidate is a three-direction selection, leading to the following construction.

\begin{proposition}[Cone-coverage escalation]
\label{prop:span}

Suppose $G\le\pi$. If $G<\pi$, then
$
    \operatorname{cone}\{a_{ij}\}_{j\in\mathcal J_i}
    =\mathbb R^2,
$
and consequently any retained subset satisfying \eqref{eq:cone} must
positively span $\mathbb R^2$.
If $G=\pi$, the full conic hull is either a line in the collinear case or a
closed half-plane otherwise. Two antiparallel rays are sufficient in the
former case, whereas at least three retained directions are required in the
latter.

For a uniform implementation, the $G\le\pi$ branch is initialized with
$K=3$ whenever $|\mathcal J_i|\ge3$. For each candidate size $K$, the
geometric selector proposes a retained set of cardinality $K$, which is
accepted only if it satisfies the exact cone-coverage condition
\eqref{eq:cone}. If cone coverage fails, $K$ is increased by one and a new
candidate is proposed. The procedure terminates no later than
$K=|\mathcal J_i|$, where the retained set is the full safety-constraint set
and \eqref{eq:cone} holds trivially.
\end{proposition}

\begin{proof}
Let
$
    C_i
    :=
    \operatorname{cone}\{a_{ij}\}_{j\in\mathcal J_i},
    $ and $
    \widetilde C_i
    :=
    \operatorname{cone}
    \{\tilde a_{ik}\}_{k\in\mathcal K_i}.
$
Because the retained directions are selected from the original set,
$
    \widetilde C_i\subseteq C_i.
$
Conversely, if \eqref{eq:cone} holds, every original generator belongs to
$\widetilde C_i$, and therefore
$
    C_i\subseteq\widetilde C_i.
$
Hence every accepted retained set satisfies
$
    \widetilde C_i=C_i.
$

If $G<\pi$, no closed half-plane contains all original normals, so
$C_i=\mathbb R^2$. Hence every accepted retained set positively spans
$\mathbb R^2$.

If $G=\pi$, $C_i$ is either a line in the collinear case or a closed
half-plane otherwise. In the latter case two rays cannot generate $C_i$, so
at least three retained directions are required.

Every failed cone-coverage test increases $K$, while
$K\le|\mathcal J_i|$. At $K=|\mathcal J_i|$, all original directions are
retained and \eqref{eq:cone} holds trivially. Hence the escalation
terminates.
\end{proof}

Propositions~\ref{prop:cone2} and~\ref{prop:span} address the selection problem -- they identify a small set of retained directions intended to satisfy the cone condition of Lemma~\ref{lem:farkas}. Containment, however, requires more than selecting suitable directions; for every dropped constraint, the corresponding nonnegative multipliers must also be constructed explicitly. In two dimensions this certification step does not require an auxiliary linear program. Once the retained directions are ordered by angle, each non-collinear covered normal can be represented using the two retained rays bounding an appropriate conic sector; a collinear normal requires only one retained ray.

\begin{proposition}[Closed-form conic multipliers in $\mathbb R^2$]
\label{prop:cf}
Let
$
    a_{ij}\in
    \operatorname{cone}
    \{\tilde a_{ik}\}_{k\in\mathcal K_i}.
$
Then $a_{ij}$ admits a conic representation using at most two retained
directions. If $a_{ij}$ lies on a retained ray $\tilde a_{ip}$, i.e., is positively
collinear with $\tilde a_{ip}$, then
$
    a_{ij}=\lambda_{jp}\tilde a_{ip},
    \; \lambda_{jp}\ge0.
$
Otherwise, there exist two retained rays $\tilde a_{ip}$ and
$\tilde a_{iq}$ whose angular separation satisfies
$0<\Delta\theta_{pq}<\pi$ and whose cone contains $a_{ij}$, such that
$
    a_{ij}
    =
    \lambda_{jp}\tilde a_{ip}
    +
    \lambda_{jq}\tilde a_{iq},
    \;
    \lambda_{jp},\lambda_{jq}\ge0.
$
The multipliers are given by
$\begin{bmatrix}
        \lambda_{jp}\\
        \lambda_{jq}
  \end{bmatrix} =
  \begin{bmatrix}
        \tilde a_{ip} & \tilde a_{iq}
  \end{bmatrix}^{-1}a_{ij},$
with $\lambda_{jk}=0$ for $k\notin\{p,q\}$.
\end{proposition}

\begin{proof}
By the conic form of Carath\'eodory's theorem, every vector in a cone in
$\mathbb R^2$ can be represented using at most two generators. If $a_{ij}$ lies on a retained ray, one generator suffices and the first
case follows. Otherwise choose two retained rays forming a sector of angular width strictly less than $\pi$ whose cone contains $a_{ij}$. Their angular separation is strictly
between $0$ and $\pi$, so they are linearly independent and the associated
$2\times2$ system has a unique solution. Since $a_{ij}$ lies in their cone,
both multipliers are nonnegative.
\end{proof}

Proposition~\ref{prop:cf} therefore converts the geometric selection of
Propositions~\ref{prop:cone2}--\ref{prop:span} into the explicit nonnegative
multipliers required by the directional condition of
Lemma~\ref{lem:farkas}. Once cone coverage holds, the remaining offset
condition is enforced constructively by the nominal-aware offset design
introduced next. Consequently, cone coverage is the only certification
failure mode at this stage; failure of cone coverage triggers escalation in
$K$.
A separate escalation may subsequently be required if the certified reduced
DMPC problem is infeasible. This is an OCP-feasibility issue rather than a
failure of the containment certificate.

\subsection{Choosing the offsets}
Direction selection determines the normals of the reduced halfspaces but not
their offsets. Here, $b_{ij}$ denotes the already tube-tightened offset of the
full eCBF constraint. Given certified multipliers satisfying
$
    a_{ij}
    =
    \sum_{k=1}^{K}\lambda_{jk}\tilde a_{ik},
    \;
    \lambda_{jk}\ge0.
$
Lemma~\ref{lem:farkas} requires the reduced offsets to satisfy
$
    \sum_{k=1}^{K}
    \lambda_{jk}\tilde b_k
    \ge b_{ij},
    \;
    j\in\mathcal J_i.
$
Among all such certified offsets, we choose those that minimize the
worst-case offset margin relative to a prescribed nominal control
$v^{\mathrm{nom}}$:
\begin{equation}\label{eq:nomaware}
  \begin{aligned}
  \min_{\tilde b,\epsilon}\quad & \epsilon\\
  \text{s.t.}\quad
  & \tilde b_k
      \le \tilde a_{ik}^{\top}v^{\mathrm{nom}}+\epsilon,
      && k=1,\ldots,K,\\
  & \sum_{k=1}^{K}
      \lambda_{jk}\tilde b_k
      \ge b_{ij},
      && j\in\mathcal J_i,\\
  & \tilde b_k\ge b_k^{\mathrm{kept}},
      && k=1,\ldots,K .
  \end{aligned}
\end{equation}
Here, $\epsilon\in\mathbb R$ is an auxiliary scalar measuring the worst-case
offset margin of the reduced constraints relative to
$v^{\mathrm{nom}}$; minimizing it selects the certified offsets that are most
favorable to the nominal control. The containment inequalities remain hard constraints; hence the offset
selection affects conservatism and nominal-control preservation, but not the
validity of the containment certificate.

\begin{proposition}[Closed-form nominal-aware offsets]
\label{prop:offset}
Let
$
    s_k:=\tilde a_{ik}^{\top}v^{\mathrm{nom}},
    \;
    d_j:=\sum_{k=1}^{K}\lambda_{jk}.
$
If $a_{ij}\neq0$, then cone coverage implies $d_j>0$. The optimal
value of \eqref{eq:nomaware} is
$
\epsilon^\star
=
\max\left\{
    \max_{1\le k\le K}
        \bigl(b_k^{\mathrm{kept}}-s_k\bigr),
    \;
    \max_{j\in\mathcal J_i}
        \frac{
            b_{ij}
            -
            \sum_{k=1}^{K}\lambda_{jk}s_k
        }{d_j}
\right\}.
$
One optimal offset vector is
$
    \tilde b^\star
    =
    s+\epsilon^\star\mathbf 1.
$
\end{proposition}

\begin{proof}
For fixed $\epsilon$, the first and third constraints of
\eqref{eq:nomaware} require
$
    b^{\mathrm{kept}}
    \le \tilde b
    \le s+\epsilon\mathbf 1 .
$
Since $\lambda_{jk}\ge0$, each containment expression
$\sum_{k=1}^{K}\lambda_{jk}\tilde b_k$ is nondecreasing componentwise
in $\tilde b$. Hence a feasible $\tilde b$ exists if and only if the
largest admissible vector,
$
    \tilde b=s+\epsilon\mathbf 1,
$
satisfies
$
    s_k+\epsilon\ge b_k^{\mathrm{kept}},
    \; k=1,\ldots,K,
$
and
$
    \sum_{k=1}^{K}\lambda_{jk}(s_k+\epsilon)
    \ge b_{ij},
    \; j\in\mathcal J_i.
$
Because $d_j=\sum_{k=1}^{K}\lambda_{jk}>0$, these inequalities are
equivalent to
$
    \epsilon\ge b_k^{\mathrm{kept}}-s_k
$
and
$
    \epsilon
    \ge
    \frac{
        b_{ij}-\sum_{k=1}^{K}\lambda_{jk}s_k
    }{d_j},
$
respectively. Taking the maximum of all lower bounds gives
$\epsilon^\star$, and
$\tilde b^\star=s+\epsilon^\star\mathbf1$ is feasible at this value.
\end{proof}

Thus, once cone coverage is established, Proposition~\ref{prop:offset}
provides finite reduced offsets satisfying the offset condition of
Lemma~\ref{lem:farkas}. Hence cone coverage is the only
certificate-construction failure mode; no separate offset-level
certification or escalation is required.

At $K=|\mathcal J_i|$, we choose the full-set fallback
$
    \tilde a_{ij}=a_{ij},
    \;
    \tilde b_{ij}=b_{ij},
$
with identity multipliers. Then
$
    \epsilon^\star
    =
    \max_{j\in\mathcal J_i}(b_{ij}-s_j),
$
so $b_{ij}\le s_j+\epsilon^\star$ for every $j$. Hence
$\tilde b=b$ is feasible at the optimal value $\epsilon^\star$ and is
therefore itself an optimal offset choice. Moreover,
$
    \sum_k\lambda_{jk}\tilde a_{ik}=a_{ij};
    \; \sum_k\lambda_{jk}\tilde b_{ik}=b_{ij},
$
so the reduced and unreduced constraint sets coincide exactly.

The selector therefore affects only the achieved reduction. Correctness is
determined by certification of the proposed retained set -- a failed directional
candidate increases the retained constraint count, while every accepted
candidate satisfies the same containment guarantee.


\section{Main Result} \label{MainResult}
Having established a constructive procedure for proposing and certifying
reduced constraint sets, we now state the guarantees inherited by the
resulting controller. The implemented construction acts directly on the
tube-tightened eCBF halfspaces. Lemma~\ref{lem:tube-transfer} additionally
shows that a nominal-level directional certificate can be transferred through
the tube tightening with the same conic multipliers.

\begin{assumption}[Tube validity]\label{as:tube}
The tube and disturbance bounds used in constructing the tightened offsets
$b_{ij}$ are valid for the true closed-loop tracking error.
\end{assumption}

Propositions~\ref{prop:cone2}--\ref{prop:span} characterize the planar
selection geometry, while Proposition~\ref{prop:cf} provides the conic
multipliers once a candidate retained set passes cone coverage.
Proposition~\ref{prop:offset} then constructs offsets satisfying the remaining
condition of Lemma~\ref{lem:farkas}.

\begin{theorem}[Certified reduction preserves robust safety]
\label{thm:main}
Let Assumptions~\ref{ass:reldeg} and~\ref{as:tube} hold. At a fixed
prediction stage, let
$
    a_{ij}^{\top}v_i\ge b_{ij},
    \; j\in\mathcal J_i,
$
we denote the full tube-tightened safety eCBF constraints. Suppose the retained
directions and nonnegative multipliers satisfy
$
    a_{ij}
    =
    \sum_{k=1}^{K}\lambda_{jk}\tilde a_{ik},
    \;
    \lambda_{jk}\ge0,
    \;
    j\in\mathcal J_i,
$
and let the reduced offsets be constructed according to
Proposition~\ref{prop:offset}. Then
$
    \widetilde{\mathcal U}_i
    \subseteq
    \mathcal U_i.
$
Consequently, every reduced-feasible control satisfies all full
tube-tightened safety eCBF constraints. If these tightened constraints are
satisfied by the applied control along the closed-loop trajectory, the tube
validity assumption holds, and the required eCBF initial conditions are
satisfied, then the corresponding original safety sets are forward invariant
for the true disturbed system.
\end{theorem}

\begin{proof}
For every $j\in\mathcal J_i$, the accepted directional certificate gives
$
    a_{ij}
    =
    \sum_{k=1}^{K}
    \lambda_{jk}\tilde a_{ik},
    \;
    \lambda_{jk}\ge0.
$
By Proposition~\ref{prop:offset},
$
    \sum_{k=1}^{K}
    \lambda_{jk}\tilde b_k
    \ge b_{ij}.
$
Hence both conditions of Lemma~\ref{lem:farkas} hold, and therefore
$
    \widetilde{\mathcal U}_i
    \subseteq
    \mathcal U_i.
$
Thus every reduced-feasible control also satisfies every full tube-tightened
eCBF constraint. Under Assumption~\ref{as:tube}, these tightened inequalities
imply the corresponding eCBF inequalities for the true disturbed state.
Together with Assumption~\ref{ass:reldeg}, the required eCBF initial
conditions, and the underlying eCBF invariance result, this yields forward
invariance of the original safety sets.
\end{proof}

\begin{corollary}[Preservation of a safe nominal control]
\label{cor:nominal}
Under the hypotheses of Proposition~\ref{prop:offset}, let
$\mathcal U_i$ and $\widetilde{\mathcal U}_i$ denote the full and reduced
tube-tightened admissible sets, respectively. If
$
    v^{\mathrm{nom}}\in\mathcal U_i,
$
then
$
    \epsilon^\star\le0,
    \;
    v^{\mathrm{nom}}\in\widetilde{\mathcal U}_i .
$
\end{corollary}

\begin{proof}
Since $v^{\mathrm{nom}}\in\mathcal U_i$, every retained original constraint
satisfies
$
    s_k=\tilde a_{ik}^{\top}v^{\mathrm{nom}}
    \ge b_k^{\mathrm{kept}},
    \; k=1,\ldots,K.
$
Moreover, using the directional certificate,
$
    \sum_{k=1}^{K}\lambda_{jk}s_k
    =
    \left(\sum_{k=1}^{K}\lambda_{jk}\tilde a_{ik}\right)^{\top}
    v^{\mathrm{nom}}
    =
    a_{ij}^{\top}v^{\mathrm{nom}}
    \ge b_{ij},
    \; j\in\mathcal J_i.
$
Since $d_j>0$, both families of terms defining $\epsilon^\star$ in
Proposition~\ref{prop:offset} are nonpositive. Hence
$\epsilon^\star\le0$. The nominal-aware constraints then give
$
    \tilde b_k
    \le s_k+\epsilon^\star
    \le s_k,
    \; k=1,\ldots,K,
$
so
$
    \tilde a_{ik}^{\top}v^{\mathrm{nom}}
    =s_k\ge\tilde b_k .
$
Therefore
$v^{\mathrm{nom}}\in\widetilde{\mathcal U}_i$.
\end{proof}

Together, Theorem~\ref{thm:main} and Corollary~\ref{cor:nominal} provide
complementary guarantees. The reduced program admits no control excluded by
the full tightened constraints, while the nominal-aware offset construction
preserves the prescribed nominal control whenever it is admissible for the
full set. Thus containment preserves safety without unnecessarily excluding
an already-safe nominal action. The containment guarantee is independent of
the rule used to propose the retained set, provided the accepted set satisfies
Lemma~\ref{lem:farkas}; with the offsets of
Proposition~\ref{prop:offset}, nominal-control preservation is likewise
independent of the proposal rule. The selector therefore affects only the
achieved reduction and computational cost.


\section{Implementation}

The reduction is applied directly to the already tube-tightened eCBF
halfspaces
$
    a_{ij}(l)^\top \bar v^i_{l|t_k}\ge b_{ij}(l),
    \; j\in\mathcal J_i.
$
Because the predicted geometry varies along the horizon, the stage-dependent
normals, conic multipliers, and reduced offsets are recomputed at every
prediction stage. The retained indices are reused whenever possible:
an initial set $\mathcal K_i(0)$ is selected at $l=0$, while for $l>0$
the previous index set is reused and re-certified using the current-stage
normals. If cone coverage fails, the retained set is reselected or enlarged
for that stage and certified again. Thus the stagewise construction follows
$
    \text{select}
    \;\rightarrow\;
    \text{reuse and re-certify}
    \;\rightarrow\;
    \text{reselect/enlarge only on directional failure}.
$
Only the retained indices are reused; the normals, multipliers, and offsets
are always recomputed from the current predicted geometry. Any subsequent
enlargement caused by reduced-OCP infeasibility is handled separately at the
OCP level. This is all summarized in Figure~\ref{fig:certified-reduction-flow} and Algorithm~\ref{alg:reduction}.

\begin{figure*}[!t]
\centering
\includegraphics[width=\textwidth]{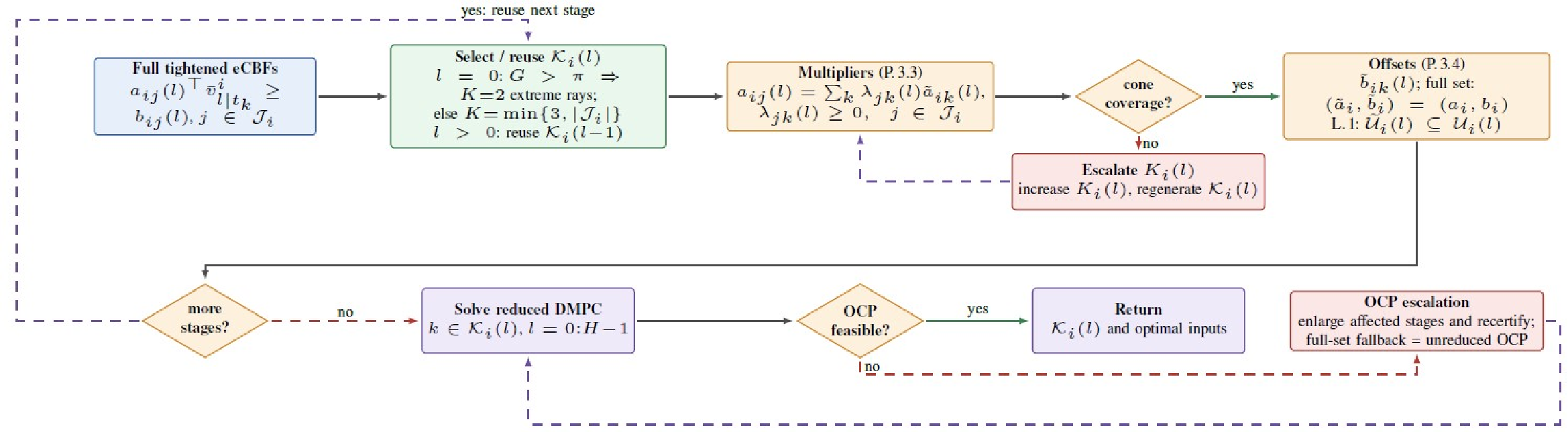}
\caption{Certified eCBF constraint-reduction workflow for agent $i$.
Each prediction stage selects or reuses a retained set, certifies cone
coverage, and constructs offsets yielding
$\widetilde{\mathcal U}_i(l)\subseteq\mathcal U_i(l)$.
Directional failure triggers stage-local escalation, while OCP infeasibility
enlarges affected stage sets up to the exact full-set fallback.}
\label{fig:certified-reduction-flow}
\end{figure*}

\begin{algorithm}[t]
\caption{Certified eCBF Constraint Reduction for Agent $i$}
\label{alg:reduction}
\small

\textbf{Data:}
Predicted states
$\{\bar x^i_{l|t_k}\}_{l=0}^{H}$,
neighbor predictions
$\{\hat x^j_{l|t_k}\}_{j\in\mathcal N_i,\;l=0}^{H}$,
obstacle predictions; nominal inputs
$\{\bar v^{i,\mathrm{nom}}_{l|t_k}\}_{l=0}^{H-1}$; full tightened eCBF
constraints \eqref{eq:lin}--\eqref{eq:bij}.
\medskip

\begin{algorithmic}[1]

\State \textbf{Construct full horizon constraints:}
for each $j\in\mathcal J_i$ and $l=0,\ldots,H-1$, compute
\[
a_{ij}(l)^\top\bar v^i_{l|t_k}\ge b_{ij}(l).
\]

\State \textbf{Initial selection:}
sort $\{a_{ij}(0)\}_{j\in\mathcal J_i}$ by angle and compute the largest
gap $G(0)$. If $G(0)>\pi$, set $K_i(0)\leftarrow2$ and retain the two
extreme rays (Prop.~\ref{prop:cone2}); otherwise set
$
    K_i(0)\leftarrow\min\{3,|\mathcal J_i|\}
$
and generate a candidate $\mathcal K_i(0)$ using the prescribed selector.
Candidate acceptance is determined by cone coverage in the stagewise
certification below.

\For{$l=0,\ldots,H-1$}

    \If{$l>0$}
        \State Reuse the preceding retained index set:
        $\mathcal K_i(l)\leftarrow\mathcal K_i(l-1)$.
    \EndIf

    \State Assign identity multipliers to retained constraints and compute,
    for every $j\in\mathcal J_i\setminus\mathcal K_i(l)$, the
    stage-dependent conic multipliers $\lambda_j(l)$ using
    Proposition~\ref{prop:cf}.

    \While{cone coverage fails and
           $|\mathcal K_i(l)|<|\mathcal J_i|$}
        \State Increase the retained-set size and generate a new candidate
        $\mathcal K_i(l)$ from the stage-$l$ geometry.
        \State Recompute the retained identity multipliers and the
        multipliers of all dropped constraints.
    \EndWhile

    \If{$\mathcal K_i(l)=\mathcal J_i$}
        \State Use the exact full-set fallback:
        $\tilde a_i(l)\leftarrow a_i(l)$,
        $\tilde b_i(l)\leftarrow b_i(l)$.
    \Else
        \State Construct $\tilde b_i(l)$ using
        Proposition~\ref{prop:offset}.
    \EndIf

\EndFor

\State Solve the DMPC OCP using
\[
    \tilde a_{ik}(l)^\top \bar v^i_{l|t_k}
    \ge
    \tilde b_{ik}(l),
    \qquad
    k\in\mathcal K_i(l),\quad
    l=0,\ldots,H-1 .
\]

\While{the reduced OCP is infeasible and some stage is not at full fallback}
    \State Enlarge one or more non-full retained sets and recertify the affected
stages; if a stage reaches the full set, restore its original normals and
offsets. Resolve the OCP.
\EndWhile

If $\mathcal K_i(l)=\mathcal J_i$ for all $l$, the reduced OCP coincides
exactly with the unreduced OCP; hence feasibility of the latter guarantees
finite recovery under full-set escalation.

\State \textbf{Return:}
$\{\bar v^{i*}_{l|t_k}\}_{l=0}^{H-1}$ and the certified retained sets, $\{\mathcal K_i(l)\}_{l=0}^{H-1}$.

\end{algorithmic}

\medskip
\textbf{Fallback:}

At any stage $l$, if
$ \mathcal K_i(l)=\mathcal J_i, $
no safety constraint is dropped. Hence
$ \tilde a_{ik}(l)=a_{ik}(l),
    \;
    \tilde b_{ik}(l)=b_{ik}(l), $
and choosing $\lambda_{jj}(l)=1$ and
$\lambda_{jk}(l)=0$ for $k\neq j$ satisfies both conditions of
Lemma~\ref{lem:farkas} with equality.
\end{algorithm}

\begin{remark}
If the certified reduced OCP is infeasible, the retained sets are enlarged
at selected prediction stages and the corresponding certificates are
recomputed. This escalation may continue up to the full-set fallback
$
    \mathcal K_i(l)=\mathcal J_i,
    \; l=0,\ldots,H-1.
$
At full fallback, the original normals and offsets are restored,
$
    \tilde a_{ij}(l)=a_{ij}(l),
    \;
    \tilde b_{ij}(l)=b_{ij}(l),
$
so no safety constraint is removed and the reduced OCP coincides exactly
with the unreduced OCP. Therefore, provided the escalation procedure is
allowed to reach this fallback, feasibility of the unreduced OCP guarantees
that the escalation terminates with a feasible problem.
\end{remark}

\section{Numerical Simulation} \label{NUMERICALEXAMPLE}
The numerical scenario extends the reference implementation
of~\cite{koulong2026cdc} to ten followers and four obstacles. We retain its
follower jerk-servo model, five desired-acceleration fields, eCBF gains
$(\kappa_0,\kappa_1,\kappa_2)=(30,38,3)$, safety distance
$d_{\mathrm{safe}}=0.1$\,m, input bound $u_{\max}=300$, and two original
obstacles. Agents $6$--$10$ reuse the five follower fields cyclically.
The present study adds two obstacles and specifies new formation offsets
$\psi^i$, a leader trajectory, sampling period $T_s=0.05$\,s, disturbance
bound $\|w^i\|_\infty\le0.4$.

\begin{figure*}[t]
    \centering
    \begin{subfigure}[b]{0.32\textwidth}
        \centering
        \includegraphics[height=1.08in]{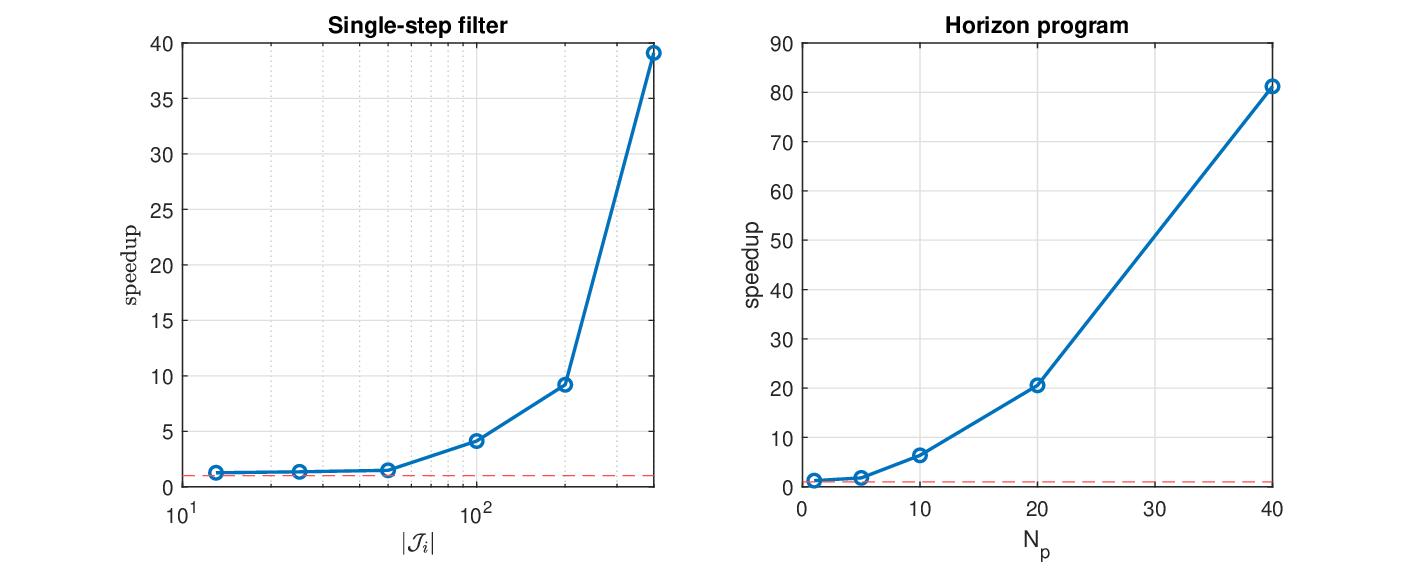}
        \caption{Computational speedup versus constraint count and horizon.}
        \label{fig:speedup_summary}
    \end{subfigure}
    \hfill
    \begin{subfigure}[b]{0.32\textwidth}
        \centering
        \includegraphics[height=1.08in]{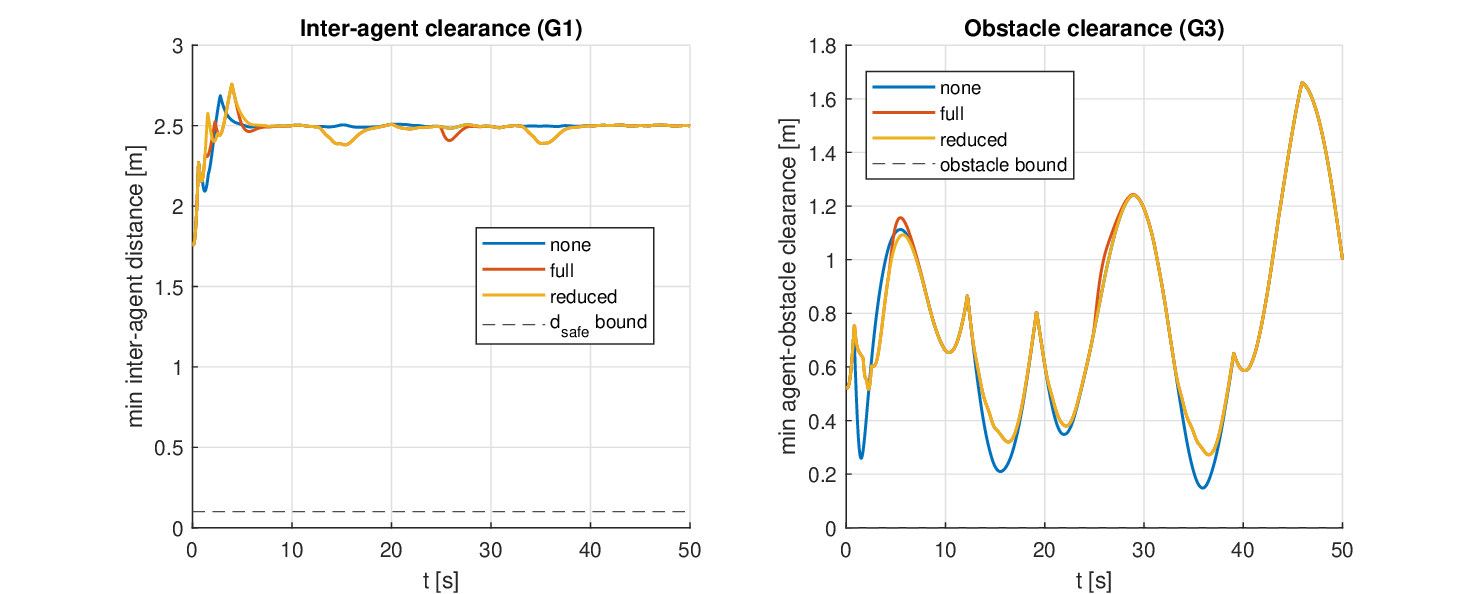}
        \caption{Inter-agent and obstacle-clearance histories.}
        \label{fig:clearance_summary}
    \end{subfigure}
    \hfill
    \begin{subfigure}[b]{0.32\textwidth}
        \centering
        \includegraphics[height=1.08in]{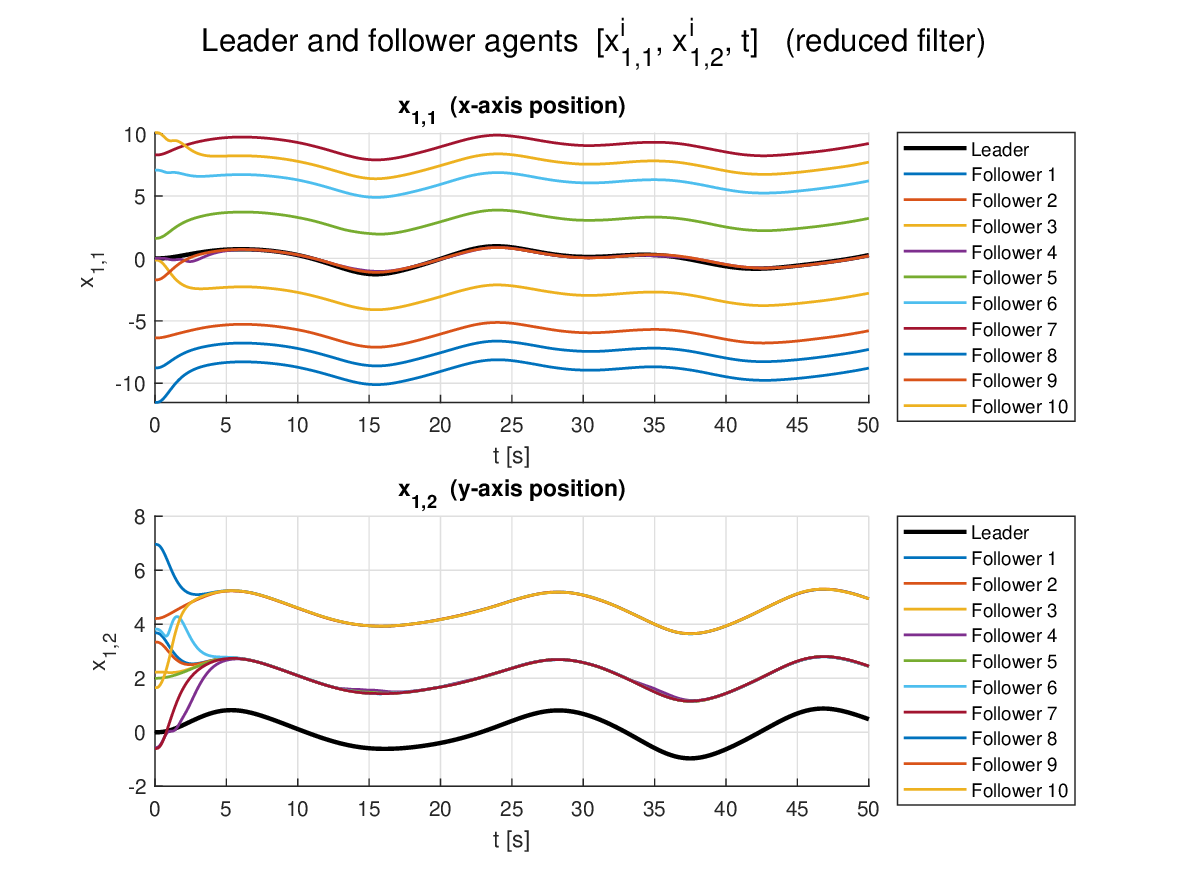}
        \caption{Leader and follower trajectories under the reduced controller.}
        \label{fig:position_summary}
    \end{subfigure}
    \caption{Computational and closed-loop results for the ten-follower,
    four-obstacle study. (a) Speedup of the reduced formulation relative to
    the full formulation as the number of safety constraints and prediction
    horizon increase. (b) Representative inter-agent and agent--obstacle
    clearances for the unfiltered, full, and reduced controllers.
    (c) Leader and follower position trajectories under the reduced controller.}
    \label{fig:numerical_summary}
\end{figure*}

\subsection{Closed-loop results}

\begin{center}
\small
\setlength{\tabcolsep}{4pt}
\begin{tabular}{lrrr}
\toprule
& no filter & full ($13$) & reduced\\
\midrule
final formation error    & $0.070722$\,m & $0.070722$\,m & $0.070722$\,m\\
min obstacle clearance   & $0.1483$\,m & $0.2721$\,m & $\mathbf{0.2721}$\,m\\
min inter-agent distance & $1.75$\,m & $1.75$\,m & $1.75$\,m\\
mean $K_i(l)$            & --- & $13.00$ & $\mathbf{2.80}$\\
nominal rejected         & --- & $2.7\%$ & $2.5\%$\\
constraint violations    & --- & $0.0\%$ & $\mathbf{0.0\%}$\\
escalated                & --- & --- & $2.4\%$\\
$G>\pi$ / $G\le\pi$      & --- & --- & $44\%$ / $56\%$\\
\bottomrule
\end{tabular}
\end{center}

The reduced formulation retains $2.80$ safety constraints on average versus
$13$ for the full formulation, a $4.6\times$ reduction. The initial selector
uses $K=2$ for $G>\pi$ and $K=3$ otherwise; only $2.4\%$ of instances require
directional escalation. Full and reduced formulations give the same reported
final formation error and minimum clearances, with no violations of the
original tightened halfspaces. Here, ``constraint violations'' is defined only
for filtered runs and denotes violation of at least one original
tube-tightened eCBF halfspace by the applied input.

Figure~\ref{fig:position_summary} shows convergence to the prescribed
formation followed by bounded-error leader tracking. In
Figure~\ref{fig:clearance_summary}, inter-agent separation remains well above
$d_{\mathrm{safe}}=0.1$\,m, while the minimum obstacle clearance increases
from $0.1483$\,m without filtering to $0.2721$\,m with either safety
formulation. Thus the unfiltered run remains geometrically safe, but the
observable filtering effect is primarily increased obstacle clearance.

\subsection{Where the reduction pays}
\label{sec:wallclock}

To expose scaling beyond the single-step closed-loop study, synthetic
clustered constraint sets vary $|\mathcal J_i|$ and $N_p$ independently.
Sweep~A fixes $N_p=1$, whereas Sweep~B fixes $|\mathcal J_i|=13$.

\begin{center}
\small
\setlength{\tabcolsep}{2pt}
\begin{tabular}{rrrr@{\hskip 1.5em}rrrr}
\toprule
\multicolumn{4}{c}{Sweep A: $N_p=1$}
& \multicolumn{4}{c}{Sweep B: $|\mathcal J_i|=13$}\\
\cmidrule(lr){1-4}\cmidrule(lr){5-8}
$|\mathcal J_i|$ & full & red. & ratio
& $N_p$ & full & red. & ratio\\
\midrule
13  & $0.871$   & $0.691$ & $1.26\times$ & 1  & $0.875$   & $0.704$ & $1.24\times$\\
25  & $0.870$   & $0.647$ & $1.35\times$ & 5  & $2.848$   & $1.591$ & $1.79\times$\\
50  & $1.409$   & $0.951$ & $1.48\times$ & 10 & $15.963$  & $2.507$ & $6.37\times$\\
100 & $6.108$   & $1.482$ & $4.12\times$ & 20 & $58.728$  & $2.855$ & $20.57\times$\\
200 & $17.420$  & $1.895$ & $9.19\times$ & 40 & $400.894$ & $4.937$ & $\mathbf{81.20\times}$\\
400 & $128.104$ & $3.276$ & $\mathbf{39.10\times}$ & & & & \\
\bottomrule
\end{tabular}
\end{center}

Times are milliseconds per solve and include all reduction overhead. Across
Sweep~A, full and reduced solve times grow by about $147\times$ and
$4.7\times$; across Sweep~B the corresponding factors are $458\times$ and
$7.0\times$. This reflects the scaling
$
N_p|\mathcal J_i|
\;\text{versus}\;
\sum_{l=0}^{N_p-1}K_i(l)\approx N_p\bar K_i .
$
At $|\mathcal J_i|=13$ and $N_p=1$, the isolated benchmark gives only a
$1.24$--$1.26\times$ speedup, while the complete closed-loop implementation
gives $0.94\times$ end-to-end. Thus reduction is essentially cost neutral for
the small single-step case but becomes increasingly beneficial as
$|\mathcal J_i|$ or $N_p$ grows.

\section{CONCLUSION}\label{CONCLUSION}
This paper presented a certified reduction framework for tube-tightened
multi-agent eCBF constraints in distributed MPC. The method exploits cone
coverage to retain only a small subset of safety halfspaces while guaranteeing
that the reduced admissible set remains contained in the full one. The resulting
controller therefore preserves the safety guarantees of the unreduced
formulation while reducing the number of explicitly enforced constraints.

Numerical results confirm that the reduction preserves the reported closed-loop
performance of the full formulation and becomes increasingly beneficial as the
number of safety constraints or prediction horizon grows. Future work will
extend the planar geometric construction to higher-dimensional input spaces.


\end{document}